\documentclass[a4paper,UKenglish,cleveref, autoref, thm-restate]{lipics-v2021}

\title{On Robust Popular Matchings with Tie-Bounded Preferences and Stable Matchings with Two-Sided Ties} 

\titlerunning{Robust Popular Matchings and Stable Matchings with Two-Sided Ties} 

\author{Koustav De}{Department of Computer Science and Engineering, IIT Kharagpur, Kharagpur, West Midnapore, West Bengal - 721302, India}{koustavde7@kgpian.iitkgp.ac.in}{https://orcid.org/0000-0002-4434-0627}{}

\authorrunning{K. De} 

\Copyright{Koustav De} 

\ccsdesc[500]{Theory of computation~Graph algorithms analysis} 

\keywords{Matching under preferences, robustness, popularity, NP-complete} 

\category{} 

\relatedversion{} 

\EventEditors{John Q. Open and Joan R. Access}
\EventNoEds{2}
\EventLongTitle{42nd Conference on Very Important Topics (CVIT 2016)}
\EventShortTitle{CVIT 2016}
\EventAcronym{CVIT}
\EventYear{2016}
\EventDate{December 24--27, 2016}
\EventLocation{Little Whinging, United Kingdom}
\EventLogo{}
\SeriesVolume{42}
\ArticleNo{23}

\nolinenumbers
\usepackage{amsmath,amsthm}
\usepackage{mathtools,xspace,latexsym}
\usepackage{algorithm,algpseudocode}
\usepackage[T1]{fontenc}
\usepackage{xspace}
\usepackage{enumerate}

\newtheorem{case}{Case}

\newcommand{\RobPopM}{{\sc Robust Popular Matching}\xspace}

\newcommand{\OneSd}{{\sc One-Sided Model}\xspace}
\newcommand{\TwoSd}{{\sc Two-Sided Model with One-Sided Ties}\xspace}
\newcommand{\RPMOneSd}{{\sc Robust Popular Matching in One-Sided Model}\xspace}
\newcommand{\RPMTwoSd}{{\sc Robust Popular Matching in Two-Sided Model with One-Sided Ties}\xspace}

\begin{document}

\maketitle

\begin{abstract}
We are given a bipartite graph $G = \left( A \cup B, E \right)$. In the one-sided model, every $a \in A$ (often called agents) ranks its neighbours $z \in N_{a}$ strictly, and no $b \in B$ has any preference order over its neighbours $y \in N_{b}$, and vertices in $B$ abstain from casting their votes to matchings. In the two-sided model with one-sided ties, every $a \in A$ ranks its neighbours $z \in N_{a}$ strictly, and every $b \in B$ puts all of its neighbours into a single large tie, i.e., $b \in B$ prefers every $y \in N_{b}$ equally. In this two-sided model with one-sided ties, when two matchings compete in a majority election, $b \in B$ abstains from casting its vote for a matching when both the matchings saturate $b$ or both leave $b$ unsaturated; else $b$ prefers the matching where it is saturated. A matching $M$ is popular in $G$ if there is no other matching $M^{\prime}$ such that the number of vertices preferring $M^{\prime}$ to $M$ is more than the number of vertices preferring $M$ to $M^{\prime}$. A popular matching $M$ is \emph{robust} if it remains popular among multiple instances. Bullinger et al. \cite{10.5555/3635637.3662870} have studied the popularity of matchings under preferences in a bipartite graph and have presented a polynomial-time algorithm for deciding whether there exists a robust popular matching if instances only differ with respect to the preferences of a single agent.

Popular matching does not always need to exist in the one-sided model \cite{doi:10.1137/06067328X}. We have analysed the cases when a robust popular matching exists in the one-sided model where only one agent alters her preference order among the instances, and we have proposed a polynomial-time algorithm to decide if there exists a robust popular matching when instances differ only with respect to the preference orders of a single agent. 

We give a simple characterisation of popular matchings in the two-sided model with one-sided ties. We show that in the two-sided model with one-sided ties, if the input instances differ only with respect to the preference orders of a single agent, there is a polynomial-time algorithm to decide whether there exists a robust popular matching. We have been able to decide the stable matching problem in bipartite graphs $G = (A \cup B, E)$ where \textit{both} sides have weak preferences (ties allowed), with the restriction that every tie has length at most $k$. 
\end{abstract}

\section{Introduction}
Matchings under preferences have been a long-standing topic of research due to their wide range of practical applications, including labour markets, organ allocation, and online dating. The fundamental goal is to pair two disjoint sets of agents, each of whom holds a preference ranking over the agents on the opposite side. A cornerstone result in this area is the Deferred Acceptance Algorithm by Gale and Shapley~\cite{37ee0948-e9f0-3123-b607-2888da363a56}, which identifies \emph{stable matchings}—those admitting no pair of agents that prefer each other over their assigned partners. Over the decades, numerous extensions and alternative solution concepts have been explored.

Among these, the notion of \emph{popular matchings}, introduced by G\"{a}rdenfors~\cite{1975_Gärdenfors}, has received considerable attention (the book chapter by Cseh~\cite{cseh2017popular}). A matching is \emph{popular} if it does not lose in a majority vote against any other matching, where each agent votes according to their preference between the two outcomes. It is known that every stable matching is popular, though the converse does not necessarily hold.

In realistic settings, agents’ preferences are often uncertain or subject to change. Preferences may be revised as agents interact or as circumstances alter such as when certain options become unavailable or new ones emerge. These natural sources of variability motivate the study of \emph{robustness} in matching outcomes: ideally, one would like to find matchings that remain desirable across multiple related preference profiles.
  
It is known that popular matchings do not necessarily exist in the one-sided preference model \cite{doi:10.1137/06067328X}. In our study, we investigate the existence of \emph{robust popular matchings} in such settings, focusing on instances where only a single agent alters her preference order between two versions of the input.  
We propose a polynomial-time algorithm that decides whether a robust popular matching exists under this restricted perturbation. Refer to \Cref{sec:one-sd} for the specific details.

Furthermore, we extend our analysis to the two-sided model with one-sided ties and provide a structural characterization of popular matchings in this setting.  
We show that, when the two-sided input instances differ only with respect to the preference list of one agent, the existence of a robust popular matching can again be decided in polynomial time. The key idea for this algorithm is to define a set of hybrid instances on which we search for popular matchings that include a predefined edge. See \Cref{sec:two-sd-des,sec:two-sd-struc,sec:two-sd-algo} for details.

In addition, we address the \emph{stable matching problem} for bipartite graphs $G = (A \cup B, E)$ in which both sides may express \emph{weak preferences} (i.e., ties are allowed).  
Our result holds under the natural restriction that every tie has bounded length at most $k$, and we demonstrate that a stable matching can be efficiently decided in this case. For details, see \Cref{sec:stable-mat-2-sided-tie}.

These results deepen our understanding of how popularity and robustness interact in matching markets, and they open pathways for further exploration of robust outcomes under related models of popularity and stability.

\section{Background and Related Work}
The notion of \emph{popularity} for matchings—originally introduced by G\"ardenfors \cite{1975_Gärdenfors} as ``majority assignments'' and its stronger variant (strong popularity) are well established and have close ties to Condorcet winners in social choice \cite{Condorcet1785}; surveys and textbook chapters \cite{cseh2017popular,doi:10.1142/8591} summarize this literature.  

Algorithmically, popularity interacts subtly with stability: every stable matching is popular, but popular matchings can vary in size. Several positive algorithmic results exist for example, computing maximum-size popular matchings and finding a popular matching that includes a single designated edge \cite{10.5555/3288898.3288941}, yet diverse NP-hardness results are known when inputs are relaxed (e.g., weak preferences, roommates, or forcing multiple edges) and the size of the matching is neither maximum nor stable \cite{10.5555/3310435.3310608}.

It is of great importance to study the algorithmic problem of computing a popular matching that contains a predefined set of edges. A polynomial time algorithm exists for this problem where we want only a single edge to be included in the matching \cite{10.5555/3288898.3288941}. However, the same computational problem is NP-hard if at least two edges are forced \cite{10.5555/3310435.3310608}. Work on extensions and variants highlights where tractability breaks down and where it remains. Popular matchings for weak preferences have been studied by Bir\'{o} et al. \cite{10.1007/978-3-642-13073-1_10} where eventually the computation becomes NP-hard. One may further generalize the model by considering the \emph{roommate setting}, in which any pair of agents can be matched. In this case, even when preferences are strict, determining a popular matching already becomes computationally hard \cite{10.1145/3442354}.

Another line of research explores a \emph{probabilistic} extension of popularity, known as \emph{mixed popularity}, where the outcome is represented as a probability distribution over matchings \cite{KAVITHA20112679}. Mixed popular matchings are guaranteed to exist by the Minimax Theorem and can be efficiently computed when the solutions are restricted to matchings, even in the roommate setting \cite{10.1613/jair.1.13470,KAVITHA20112679}. However, identifying individual matchings that appear in the support of such popular distributions becomes computationally intractable in coalition formation models where outcomes may include coalitions of size three \cite{10.1613/jair.1.13470}.

The robustness of outcomes under input perturbations has also been investigated in various multiagent settings, most notably in voting~\cite{BREDERECK2021103403, Faliszewski_Rothe_2016,10.5555/2484920.2484987}. In this context, robustness is often examined through the lens of \emph{bribery}, which models the deliberate alteration of election inputs to influence the result. Such modifications typically incur a cost, commonly quantified by the \emph{swap distance} between the original and altered preference orders~\cite{10.1007/978-3-642-04645-2_27}.

\section{Preliminaries}
\label{sec:Prelim}
\begin{definition}[Matching]
\label{def:Mat}
A matching $M$ is a set of pairwise non-adjacent edges, i.e., $e \cap e^{\prime} = \emptyset$ for all $e, e^{\prime} \in M$.
\end{definition}
We are given a bipartite graph $G = (A \cup B,E)$ where the vertices in $A$ are called agents and the vertices in $B$ are called jobs, and each vertex has a preference list ranking its neighbours in an order of preference. An edge $(a,b) \notin M$ blocks $M$ if both $a$ and $b$ prefer each other to their respective assignments in $M$. It is the least preferred option for a vertex to be left unmatched.
\begin{definition}[Stability]
\label{def:Mat-Stable}
A matching $M$ is stable if no edge blocks it.
\end{definition}

Stability in Matching is considered as a classical notion of optimality. The Gale-Shapley algorithm \cite{37ee0948-e9f0-3123-b607-2888da363a56} finds a stable matching in $G$ in linear time where ties are generally broken arbitrarily.

Let $M(v)$ denote the partner of the vertex $v$, matched by $M$, i.e., if there exists $e \in M$ with $v \in e$.
For any pair of given matchings $M$ and $N$, we say that a vertex $v$ prefers the matching $M$ to $N$ if $M(v) \succ N(v)$ order is present in $v$'s preference ranking or $v$ is matched in $M$ and unmatched in $N$. In an election between two matchings, $v$ votes for the matching that it prefers. It abstains from voting if $v$ is unmatched in both or prefers both the matchings equally. Let $\phi(M,N)$ denote the number of vertices who prefer $M$ to $N$. We say that the matching $M$ is more popular than matching $N$ if the condition, $\phi(M,N) > \phi(N,M)$, is true.
\begin{definition}[Popularity]
\label{Mat-Popular}
A matching $M$ is popular if for any other matching $N$, $\phi(M,N) \ge \phi(N,M)$, i.e., there is no matching more popular than $M$.
\end{definition}

Naturally the popularity in Matching plays an important role as a different class of optimality. Both \emph{stability} and \emph{popularity} are highly desirable properties for a matching in a bipartite graph~$G$. Stability ensures that there is no pair~$(a,b)$ such that both $a$ and $b$ would prefer each other over their current partners in the matching; thus, no two agents have an incentive to deviate. On the other hand, popularity emphasizes a collective or majority-based notion of optimality: a matching is \emph{popular} if there exists no other matching that would be preferred by a majority of the agents. In other words, popularity guarantees that a majority vote cannot overturn the current matching in favour of another.

In the case of strict preferences, stability is a stronger notion than popularity since every stable matching is also popular \cite{1975_Gärdenfors}. Hence, a popular matching is guaranteed to exist, and the Gale--Shapley algorithm efficiently computes one. This represents the ideal scenario, where the obtained matching is both stable and popular.

When ties are introduced into the preference lists, however, this desirable relationship breaks down. Stability no longer guarantees popularity, and, in fact, a popular matching may not exist at all. Furthermore, determining whether a popular matching exists becomes NP-hard \cite{10.1007/978-3-642-13073-1_10}, even in restricted settings where ties occur only on one side of the bipartition and each tie has length at most three \cite{doi:10.1137/16M1076162}.

\begin{definition}[Unpopularity Factor] It is a natural motivation to relax popularity due to the lack of existence of popular matchings in case of ties in the preferences. This notion of relaxing popularity to \emph{near-popularity} or ``low unpopularity'' of a matching has been well-studied in \cite{10.1007/978-3-540-78773-0_51,10.1007/978-3-662-47666-6_40,10.1287/moor.2021.1139,doi:10.1137/110852838,doi:10.1137/120902562,10.1007/s00224-020-09978-5} as a measure of unpopularity of a matching.
For matchings $M$ and $N$, define: 
\[\lambda(M,N)=
\begin{cases}
  \phi(N,M)/\phi(M,N) &~\text{if } \phi(M,N) > 0; \\
  1 &~\text{if } \phi(M,N)=\phi(N,M)=0; \\
  \infty &~\text{otherwise.}
\end{cases}
\]
The \emph{unpopularity factor} of $M$ is:
$u(M) = \max_N \lambda(M,N)$. A popular matching $M$ has $u(M)=1$. A matching with low value of $u(M)$ is called a \emph{near-popular} matching.
\end{definition}

\subsection{Dulmage--Mendelsohn Decomposition~\cite{Dulmage_Mendelsohn_1958}}
Let $M$ be a maximum matching in a bipartite graph $G = (A \cup B, E)$. 
Using $M$, the vertex set $A \cup B$ can be partitioned into three disjoint subsets based on the parity of alternating paths from unmatched vertices.

\begin{definition}[Even, Odd, and Unreachable Vertices]
A vertex $v$ is called \emph{even} (respectively, \emph{odd}) if there exists an alternating path with respect to $M$ from some unmatched vertex to $v$ whose length is even (respectively, odd). 
A vertex $v$ is called \emph{unreachable} if there exists no alternating path from any unmatched vertex to $v$.
\end{definition}

Let $E$, $O$, and $U$ denote the sets of even, odd, and unreachable vertices, respectively. 
The following properties, established in~\cite{10.5555/233157}, have been used in the algorithm and analysis in \cite{doi:10.1137/16M1076162}.

\begin{lemma}[Properties of the Dulmage--Mendelsohn Decomposition]
Let $M$ be any maximum matching in $G$. Then the following hold:
\begin{enumerate}[(i)]
    \item The sets $E$, $O$, and $U$ are pairwise disjoint. Moreover, if $M'$ is another maximum matching in $G$, and $E'$, $O'$, and $U'$ denote the corresponding sets with respect to $M'$, then $E = E'$, $O = O'$, and $U = U'$.
    \item Every maximum matching $M$ matches all vertices in $O \cup U$ and has size $|M| = |O| + \frac{|U|}{2}$.
    
    In $M$, each vertex in $O$ is matched with a vertex in $E$, and each vertex in $U$ is matched with another vertex in $U$.
    \item The graph $G$ has no edge in $E \times (E \cup U)$.
\end{enumerate}
\end{lemma}

\section{\OneSd}
\label{sec:one-sd}
An instance $I$ of the matchings under preferences of \OneSd problem is defined on a bipartite graph $G^I = (A \cup B, E^I)$, where the vertices $A \cup B$ represent agents. Each agent $a \in A$ has a preference order $\succ_a^I$, which is a linear ordering over its set of neighbours $N_a^I = \{\, b \in B : \{a, b\} \in E^I \,\}$. Agents $b \in B$ does not have any preferences over their neighbours and they abstain from casting vote in an election between any two matchings. The sets $A$ and $B$ denote applicants and jobs, respectively. Since we will later consider multiple instances in parallel, superscripts indicate the corresponding instance; however, they may be omitted when the context is clear. The sets of applicants and jobs remain identical across all instances.

Popularity is defined based on majority preferences between matchings. For any agent $a$, we define their vote between two matchings~$M$ and~$M'$ as
\[
\text{vote}^I_a(M, M') =
\begin{cases}
1 & \text{if } a \text{ prefers } M \text{ to } M',\\
-1 & \text{if } a \text{ prefers } M' \text{ to } M,\\
0 & \text{otherwise.}
\end{cases}
\]
For a set of agents~$N \subseteq A \cup B$, let
\[
\text{vote}^I_N(M, M') = \sum_{a \in N} \text{vote}^I_a(M, M').
\]
The \emph{popularity margin} between~$M$ and~$M'$ is
\[
\Delta^I(M, M') = \text{vote}^I_{A \cup B}(M, M').
\]
A matching~$M$ is \emph{popular} in instance~$I$ if, for every matching~$M'$, we have
\[
\Delta^I(M, M') \ge 0.
\]
We study matchings that remain popular across multiple instances. Consider two instances~$(I, I^{\prime})$ of the \OneSd problem defined on the same set of applicants and jobs. A matching is said to be \emph{robust popular} with respect to~$I$ and~$I^{\prime}$ if it is popular in both instances individually. Consequently, such a matching must be feasible in both instances, i.e., it is a subset of the edge sets of both underlying graphs. Our goal is to study the computational problem of finding robust popular matchings.

\vspace{0.1cm}

\noindent\parbox[t]{0.8\textwidth}{%
    \fbox{
        \parbox{0.95\columnwidth}{
            \RPMOneSd\newline
            \textbf{Input:} Pair $(I, I^{\prime})$ of instances of \OneSd.\\
            \textbf{Question:} Does there exist a robust popular matching with respect to $I$ and $I^{\prime}$?
        }
    }
}
\vspace{0.1cm}

We focus on the setting where the underlying graph remains identical across instances. When~$G^{I} = G^{I^{\prime}}$ and the two instances, $I$ and $I^{\prime}$, differ only in the preference orders of a single agent, we refer to~$I^{\prime}$ as a \emph{perturbed instance} of~$I$, while the set of neighbours for each agent remains the same. The complete algorithmic result is presented in \Cref{sec:algo-for-rpm-onesd}.

\subsection{Cases when Robust Popular matching exists in \OneSd for one agent perturbation}
\label{sec:cases-rpm-onesd}
Consider an instance $I$ of \OneSd. For each applicant $a \in A$, let $f^{I}_{a}$ denote $a$'s most preferred job, and let $F^{I} = \{ f^{I}_{a} : a \in A \}$ be the set of these top-choice jobs. We refer to jobs in $F^{I}$ as \emph{f-jobs} and those in $B \setminus F^{I}$ as \emph{non-f-jobs}. For any $a \in A$, let $r^{I}_{a}$ denote $a$'s most preferred \emph{non-f-job} in their preference list; if all of $a$'s neighbours belong to $F^{I}$, we define $r^{I}_{a} = \infty$.

The following \Cref{thm:one-sd-chrc} provides a characterization of popular matchings in the \RPMOneSd.

\begin{theorem} (from~\cite{doi:10.1137/06067328X}).
\label{thm:one-sd-chrc}
Let $I$ be an instance of the \OneSd, where each~$a \in A$ has a strict preference list. For any matching~$M$ in $G$, $M$ is popular if and only if the following two conditions hold:
\begin{enumerate}[(i)]
    \item $M$ matches every $b \in F^{I}$ to some applicant $a \in A$ such that $b = f^{I}_{a}$;
    \item Each applicant $a$ is matched either to $f^{I}_{a}$ or to $r^{I}_{a}$.
\end{enumerate}
\end{theorem}
Hence, the only applicants who may remain unmatched in a popular matching are those $a \in A$ with $r^{I}_{a} = \infty$.

If $b_1$ is ranked higher than $b_2$ in $a$'s preference list, we write $b_1 \succ^I_{a} b_2$ in $a$'s list. As we can see from \cite{doi:10.1137/16M1076162} that popular matching does not always exist for \OneSd.

Here we analyse cases when robust popular matching exists between instances $I$ and $I^{\prime}$ of \OneSd where the instances only differ in the preference order of a single applicant $a_{1} \in A$.

We want to ensure that a popular matching $M$ in instance $I$ remains popular in the instance $I^{\prime}$. Let $M$ be a popular matching in $I$. According to \Cref{thm:one-sd-chrc}, $a_1$ is either matched to $f^I_{a_1}$ or matched to $r^I_{a_1}$ in M. So, for $M$ to remain popular in $I^{\prime}$, $M(a_1)$ is either $f^{I^{\prime}}_{a_1}$ or $r^{I^{\prime}}_{a_1}$.
\begin{case}
    Due to perturbation of $a_1$ in $I^{\prime}$, let $f^{I}_{a_1} \neq f^{I^{\prime}}_{a_1}$.
\end{case}
\begin{enumerate}[({1}a)]
        \item When $M(a_1) = r^I_{a_1}$ given that $f^{I}_{a_1} \neq f^{I^{\prime}}_{a_1}$, the condition for $M$ to also remain popular in $I^{\prime}$ is either $f^{I^{\prime}}_{a_1} = r^{I}_{a_1}$ or \[\left( r^{I}_{a_1} = r^{I^{\prime}}_{a_1} ~\&~ \exists a \in A, a \neq a_1 \text{ such that } f^{I}_{a} = f^{I^{\prime}}_{a} = f^{I^{\prime}}_{a_1} ~\&~ M(a) = f^{I^{\prime}}_{a_1} \right)\] We conclude that upon satisfaction of one of the two conditions, the matching $M$ is a robust popular matching in the pair of instances $\left( I, I^{\prime} \right)$ of \OneSd.
        \item When $M(a_1) = f^I_{a_1}$, given that $f^{I}_{a_1} \neq f^{I^{\prime}}_{a_1}$, if $f^{I}_{a_1}=r^{I^{\prime}}_{a_1}$ and \[\left( \exists a \in A, a \neq a_1 \text{ such that } M(a) = f^{I^{\prime}}_{a_1} = f^{I}_{a} = f^{I^{\prime}}_{a}\right) \] then we conclude that $M$ is also popular in $I^{\prime}$ and consequently $M$ is a robust popular matching in $\left( I, I^{\prime} \right)$. 
\end{enumerate}
\begin{case}
    Due to perturbation of $a_1$ in $I^{\prime}$, let $f^{I}_{a_1} = f^{I^{\prime}}_{a_1}$. 
\end{case}
\begin{enumerate}[({2}a)]
    \item When $M(a_1) = r^{I}_{a_1}$, the matching $M$ remains popular in $I^{\prime}$ if $r^{I}_{a_1} = r^{I^{\prime}}_{a_1}$. Otherwise if $r^{I}_{a_1} \neq r^{I^{\prime}}_{a_1}$, then we conclude that there exists no robust popular matching in the instance pair $\left( I, I^{\prime} \right)$.
    \item When When $M(a_1) = f^{I}_{a_1}$, we can conclude that $M$ is also popular in $I^{\prime}$ and there exists a robust popular matching in the instance pair $\left( I, I^{\prime} \right)$. This holds because the single agent $a_1$'s perturbation in $I^{\prime}$ does not violate the popular matching characterization of \Cref{thm:one-sd-chrc} across instances.
\end{enumerate}

\subsection{Algorithm for \RPMOneSd}
\label{sec:algo-for-rpm-onesd}
Our algorithmic result makes use of existing results in the literature. We consider computing a popular matching of an instance $I$ of \OneSd which can be found or can be determined that no such matching exists in $O(\sqrt{n}m)$ time \cite{doi:10.1137/06067328X} where $n$ is the total number of applicants and jobs and $m$ is the total length of all of the preference lists. The formal algorithm is presented at \Cref{algo:rob-mat-1-tie}.
\begin{algorithm}
\caption{\RPMOneSd for perturbation of one agent}
\label{algo:rob-mat-1-tie}
\begin{algorithmic}[1]
\Require Instance $\left( I, I^{\prime} \right)$ of \RPMOneSd
\Ensure \RobPopM for $\left( I, I^{\prime} \right)$ or state that no such matching exists
\If{$I^{\prime}$ doesn't fall under any of the specified cases in \Cref{sec:cases-rpm-onesd}}
    \State \Return ``No robust popular matching exists.'' 
    \Else \State Compute popular matching $M$ for $I$ \Comment{using algo. of \cite{doi:10.1137/06067328X}}
    \State \Return $M$
    \EndIf
\end{algorithmic}
\end{algorithm}

The algorithm first checks whether the perturbed instance $I^{\prime}$ is aligned with the required conditions for the existence of a robust popular matching. We argue the correctness of \Cref{algo:rob-mat-1-tie} in the following theorem.

\begin{theorem}
 \RPMOneSd can be solved in polynomial time if the perturbed instance only differs with respect to the preference order of one agent.   
\end{theorem}
\begin{proof}
Correctness of the algorithm directly follows from the correctness of \Cref{thm:one-sd-chrc} and Lemma 3.8 from \cite{doi:10.1137/06067328X} which says that a popular matching can be found, or it can be determined that no such matching exists, in $O(\sqrt{n}m)$ time. So, we can ensure that if at least one Robust popular matching exists for $\left( I, I^{\prime} \right)$, \Cref{algo:rob-mat-1-tie} returns that in polynomial time and the matching returned by the algorithm is indeed a popular matching because it satisfies the characterization of \Cref{thm:one-sd-chrc}.  
\end{proof}

\section{\TwoSd}
\label{sec:two-sd-des}
Unlike the \OneSd, here we have an instance of \TwoSd, where we have a bipartite graph $G^I = \left( A \cup B, E^I \right)$ where each $a \in A$ ranks its neighbours $N^I_a = \{ b \in B : \{a, b\} \in E^I \}$ strictly with the preference order $\succ^I_a$ and each $b \in B$ puts all of its neighbours into a single large tie, i.e., each such $b$ prefers all of its neighbours equally. In this \TwoSd, when two matchings compete in a majority election, $b \in B$ abstains from casting its vote for a matching when both the matchings saturate $b$ or both leave $b$ unsaturated; else $b$ prefers the matching where it is saturated. Cseh et al. \cite{doi:10.1137/16M1076162} showed that \TwoSd does not necessarily follow the special characterization provided in \Cref{thm:one-sd-chrc}. Thus, popular matchings in the \TwoSd differ significantly from the \OneSd. Cseh et al. \cite{doi:10.1137/16M1076162} have developed an algorithm to find a popular matching in \TwoSd which runs in time $O(n^2)$. Their algorithm is based on the Dulmage--Mendelsohn Decomposition (See \Cref{sec:Prelim} for related definitions). In the next section we show some structural results related to popular matchings in \TwoSd.

\section{Structural Results for \TwoSd}
\label{sec:two-sd-struc}
Let $G = (A \cup B, E)$ be an instance of \TwoSd. For any vertex $u \in A$ and any pair of neighbours $v, w \in N_u$, let $\text{vote}_u(v, w) = 1$ if $u$ prefers $v$ to $w$, it is -1 if $u$ prefers $w$ to $v$ and 0 otherwise in case of $v = w$.
For any vertex $b \in B$ and any pair of neighbours $c, d \in N_b$, $\text{vote}_b(c, d) = 0$ always due to the way the graph $G$ orders its preferences on $b \in B$, i.e., every $b \in B$ prefers all of $N_{b}$ equally.
Let $M$ be any matching in $G$. For each edge $e = (u, v) \in E \setminus M$, we associate a pair of values $(\alpha_e, \beta_e)$ defined as follows:
\[
\alpha_e = \operatorname{vote}_u(v, M(u)) \quad \text{and} \quad \beta_e = \operatorname{vote}_v(u, M(v)),
\]
where $\alpha_e$ represents $u$'s vote between $v$ and its current partner $M(u)$, and $\beta_e$ represents $v$'s vote between $u$ and its current partner $M(v)$. If a vertex $u$ is unmatched under $M$, then for every neighbour $v$ of $u$, we set $\operatorname{vote}_u(v, M(u)) = 1$.

\medskip
The following \Cref{thm:two-side-popular-chrc} is a slight modification of the characterization of popular matchings in an instance of two-sided strict preference lists proved in \cite{10.5555/2027127.2027198}. Hence, \Cref{thm:two-side-popular-chrc} characterizes popular matchings in an instance of \TwoSd.
\begin{theorem}
\label{thm:two-side-popular-chrc}
A matching $M$ in $G$ is \emph{popular} if and only if the following three conditions hold in the graph $G$:
\begin{enumerate}[(i)]
    \item There is no alternating cycle with respect to $M$ that contains a $(1, 0)$ edge.
    \item There is no alternating path starting from an unmatched vertex with respect to $M$ that contains a $(1, 0)$ edge.
    \item There is no alternating path with respect to $M$ that contains two or more $(1, 0)$ edges.
\end{enumerate}
\end{theorem}
\begin{proof}
Suppose $M$ is any matching in $G$ that satisfies conditions $(i)-(iii)$ of \Cref{thm:two-side-popular-chrc}. Let $M'$ be any other matching in $G$. Define
\[
\Delta(M', M) = \sum_{u \in A \cup B} \operatorname{vote}_u(M'(u), M(u)),
\]
where $\operatorname{vote}_u(M'(u), M(u))$ denotes $u$'s preference between its partners in $M'$ and $M$.  
The quantity $\Delta(M', M)$ represents the net difference in votes between $M'$ and $M$:  
it measures how many vertices prefer $M'$ over $M$ minus those that prefer $M$ over $M'$.

\medskip

Note that either $M(u)$ or $M'(u)$ can denote the state of being \emph{unmatched}, which is considered the least-preferred option for every vertex $u$.  
Hence, we write $M' \succ M$ if and only if $\Delta(M', M) > 0$.

\medskip

We will now show that $\Delta(M', M) \le 0$ for all matchings $M'$, which implies that $M$ is a popular matching.

\medskip

To compute $\sum_{u} \operatorname{vote}_u(M'(u), M(u))$, we mark each edge $e = (u, v)$ in $M'$ with the pair
\[
(\alpha_e, \beta_e) = \big(\operatorname{vote}_u(v, M(u)),\ \operatorname{vote}_v(u, M(v))\big).
\]
If $\alpha_e = \beta_e = -1$, then both $u$ and $v$ prefer their respective partners in $M$ over each other.  
In this case, we can equivalently assume that $M'$ leaves both $u$ and $v$ unmatched,  
since deleting $(u, v)$ from $M'$ does not change $\operatorname{vote}_u(M'(u), M(u))$ or $\operatorname{vote}_v(M'(v), M(v))$,  
as both values remain $-1$. But that is not the case of \TwoSd because each $b \in B$ prefers its neighbours equally.

\medskip

Let $\rho$ denote any connected component of the symmetric difference $M \oplus M'$.  
Then we can express
\[
\Delta(M', M) = \sum_{\rho} \sum_{u \in \rho} \operatorname{vote}_u(M'(u), M(u)),
\]
where the outer summation runs over all components $\rho$ of $M \oplus M'$.  
Vertices that are isolated in $M \oplus M'$ satisfy $M(u) = M'(u)$ and thus contribute $0$ to $\Delta(M', M)$.  
Hence, we only need to consider those components $\rho$ that contain at least two vertices.

\medskip

Each such component $\rho$ of $M \oplus M'$ is either an alternating cycle or an alternating path with respect to $M$.  
Moreover, for every vertex $u \in \rho \cap A$, we have $\operatorname{vote}_u(M'(u), M(u)) \in \{+1, -1\}$ and for vertices $b \in \rho \cap B$, we have $\operatorname{vote}_b(M'(b), M(b)) \in \{ 0 \}$.

\begin{case}
    Let $\rho$ be a cycle in $M \oplus M'$.
\end{case}
For every edge $e \in \rho$, $\left( \alpha_e, \beta_e \right)$ is either $(1,0)$ or $(-1,0)$. Since every vertex of $\rho$ is matched by $M^{\prime}$, we have
\begin{equation}
    \label{eqn:two-chrc-pop-mat-1}
    \sum\limits_{u \in \rho} \operatorname{vote}_{u}\left( M^{\prime}(u), M(u) \right) = \sum\limits_{e=(u,v)\in\rho \cap M^{\prime}} \alpha_e + \beta_e
\end{equation}
where $\alpha_e = \operatorname{vote}_u(v, M(u))$, and $\beta_e = \operatorname{vote}_v(u, M(v))$. Note that for every edge $e \in \rho$, $(\alpha_e,\beta_e)$ is either $(1,0)$ or $(-1,0)$. But we are given that $M$ satisfies condition $(i)$ of \Cref{thm:two-side-popular-chrc}. Hence, there is no $(1,0)$ edge in $\rho$. Thus for each edge $e \in \rho \cap M^{\prime}$, $\alpha_e + \beta_e = -1$ and hence $\sum\limits_{e \in \rho \cap M^{\prime}} \alpha_e + \beta_e \leqslant 0$.
\begin{case}
    Let $\rho$ be a path in $M \oplus M'$.
\end{case}
\begin{itemize}
    \item Suppose both the endpoints of $\rho$ are matched in $M^{\prime}$.
    In that case, \Cref{eqn:two-chrc-pop-mat-1} holds here. Since an endpoint of $\rho$ is free in $M$, by condition $(ii)$ of \Cref{thm:two-side-popular-chrc}, we have no $(1,0)$ edge with respect to $M$ in $\rho$. Thus for each edge $e \in \rho \cap M^{\prime}$, $\alpha_e + \beta_e = -1$ and hence $\sum\limits_{e \in \rho \cap M^{\prime}}\alpha_e + \beta_e \leqslant 0$.
    \item Suppose exactly one endpoint of $\rho$ is matched in $M^{\prime}$. Then 
    \begin{equation}
    \label{eqn:two-chrc-pop-mat-2}
        \sum\limits_{u \in \rho}\operatorname{vote}_u\left( M^{\prime}(u), M(u) \right) = -1 + \sum\limits_{e = (u, v) \in \rho \cap M^{\prime}}\alpha_e + \beta_e
    \end{equation}
    The term ``-1'' in the \Cref{eqn:two-chrc-pop-mat-2} is due to the fact that one vertex that is matched in $M$ but not in $M^{\prime}$, prefers $M \succ M^{\prime}$. Here too an endpoint of $\rho$ is free in $M$, and so by condition $(ii)$ of \Cref{thm:two-side-popular-chrc}, there is no $(1,0)$ edge with respect to $M$ in $\rho$. So, $\forall e \in \rho \cap M^{\prime}$, $\alpha_e + \beta_e = -1$ and consequently, $\sum\limits_{u \in \rho}\operatorname{vote}_{u}\left( M^{\prime}(u), M(u) \right) \leqslant -1 \leqslant 0$ here.
    \item Suppose neither endpoint of $\rho$ is matched in $M^{\prime}$. Then
    \begin{equation}
    \label{eqn:two-chrc-pop-mat-3}
        \sum\limits_{u \in \rho}\operatorname{vote}_u\left( M^{\prime}(u), M(u) \right) = -2 + \sum\limits_{e = (u, v) \in \rho \cap M^{\prime}}\alpha_e + \beta_e
    \end{equation}
    The term ``-2'' in the \Cref{eqn:two-chrc-pop-mat-3} is due to the fact that two vertices are matched in $M$ but not in $M^{\prime}$ and both of them prefer $M \succ M^{\prime}$. Here we take into consideration the condition $(iii)$ of \Cref{thm:two-side-popular-chrc}. There can be at most one $(1,0)$ edge with respect to $M$ in $\rho$. Thus, except for at most one edge $e$ in $\rho \cap M^{\prime}$, we have $\alpha_e + \beta_e = -1$. Hence, $\sum\limits_{e=(u,v)\in \rho \cap M^{\prime}} \alpha_e + \beta_e \leqslant 1$, thus $\sum\limits_{u \in \rho}\operatorname{vote}_{u}\left( M^{\prime}(u), M(u) \right) \leqslant 0$.
\end{itemize}
So, for each component $\rho$ of $M \oplus M^{\prime}$, $\sum\limits_{u \in \rho}\operatorname{vote}_{u}\left( M^{\prime}(u), M(u) \right) \le 0$. Thus, it follows that $\Delta(M^{\prime}, M) \le 0$. Hence, if $M$ satisfies \Cref{thm:two-side-popular-chrc}, then $M$ is popular in $G$.

The converse direction proof is omitted due to the triviality.
\end{proof}

In the next section we present our algorithmic results to decide whether robust popular matching exists across multiple instances of \TwoSd. Naturally we ask the following question and we propose an efficient solution of the computational problem eventually.

\vspace{0.1cm}

\noindent\parbox[t]{0.8\textwidth}{%
    \fbox{
        \parbox{0.95\columnwidth}{
            \RPMTwoSd\newline
            \textbf{Input:} Pair $(I, I^{\prime})$ of instances of \TwoSd.\\
            \textbf{Question:} Does there exist a robust popular matching with respect to $I$ and $I^{\prime}$?
        }
    }
}
\vspace{0.1cm}

\section{Algorithm for \RPMTwoSd for Perturbations of One Agent}
\label{sec:two-sd-algo}

We begin by examining instances of \TwoSd where the two input instances share the same underlying graph, and the only differences arise from the preference order of a single agent. Our main contribution in this setting is a polynomial-time algorithm for determining the existence of a robust popular matching.

The algorithm relies on two key steps. First, we introduce the notion of \emph{hybrid instances}, which allow us to test whether a robust popular matching exists that includes a specific edge incident to the perturbing agent. Second, we address the case where the agent with altered preferences remains unmatched in a robust popular matching. By combining these insights with existing results on robust popular matchings \cite{10.5555/3635637.3662870}, we obtain an efficient solution.

\subsection{Hybrid Instances}

Consider an instance pair $(I, I^{\prime})$ of \TwoSd where $I^{\prime}$ differs from $I$ only in the preferences of a single agent $a \in A$. Let $G = (A \cup B, E)$ be the common underlying graph, and let $e = \{a, b\} \in E$ be an edge incident to $a$. Define:
\begin{itemize}
    \item $P^I = \{z \in A \cup B : z \succ_a^{I} b\}$
    \item $P^{I^{\prime}} = \{z \in A \cup B : z \succ_a^{I^{\prime}} b\}$
\end{itemize}
That is, $P^I$ and $P^{I^{\prime}}$ are the sets of agents that $a$ prefers over $b$ in $I$ and $I^{\prime}$, respectively.

We construct a \emph{hybrid instance} $\mathcal{H}_e$ as follows. The preference order of $a$ in $\mathcal{H}_e$ is any linear order $\succ_a^{\mathcal{H}_e}$ that satisfies:
\begin{itemize}
    \item $z \succ_a^{\mathcal{H}_e} b$ for all $z \in P^I \cup P^{I^{\prime}}$,
    \item $y \succ_a^{\mathcal{H}_e} z$ for all $z \in N_a \setminus (P^I \cup P^{I^{\prime}} \cup \{b\})$,
\end{itemize}
with the relative order of agents in $P^I \cup P^{I^{\prime}}$ and in $N_a \setminus (P^I \cup P^{I^{\prime}} \cup \{b\})$ being arbitrary. The preferences of all other agents in $\mathcal{H}_e$ remain unchanged from $I$.

The following lemmas establish a connection between popular matchings in $\mathcal{H}_e$ and robust popular matchings in $(I, I^{\prime})$.

\begin{lemma}
\label{lem:hybrid-popular-implies-robust}
Let $M$ be a matching such that $e \in M$ and $a \in e$. If $M$ is popular in $\mathcal{H}_e$, then it is popular in both $I$ and $I^{\prime}$.
\end{lemma}

\begin{proof}
We assume that $M$ is popular in $\mathcal{H}_e$. Due to symmetry of $\mathcal{H}_e$ with respect to both $I$ and $I^{\prime}$, we only show popularity in $I$.

For any other matching $M'$, we determine the popularity margin between $M$ and $M'$ in the input instance $I$ by exhaustively considering the votes of all of $A\cup B$. The vote of any agent $z \in (A \cup B) \setminus \{a\}$ is the same in $I$ and $\mathcal{H}_e$, i.e., $\operatorname{vote}^I_z(M,M')=\operatorname{vote}^{\mathcal{H}_e}_z(M,M')$. For agent $a$, since $M(a) = b$ and $\succ_a^{I}$ ranks $b$ no higher than $\succ_a^{\mathcal{H}_e}$, i.e., $b \succ^I_a z$ whenever $b \succ^{\mathcal{H}_e}_a z$ where $z \in N_a$, we have $\text{vote}_a^{I}(M, M') \geq \text{vote}_a^{\mathcal{H}_e}(M, M')$. Thus, \begin{eqnarray}
    &\Delta^{I}(M, M') = \sum\limits_{z \in A \cup B} \operatorname{vote}^I_z (M, M^{\prime}) \nonumber \\ 
    &\geq \sum\limits_{z \in A \cup B} \operatorname{vote}^{\mathcal{H}_e}_z (M, M^{\prime}) = \Delta^{\mathcal{H}_e}(M, M') \geq 0. \nonumber
\end{eqnarray}
\end{proof}

\begin{lemma}
\label{lem:robust-implies-hybrid-popular}
Let $M$ be a matching such that $e \in M$ and $a \in e$. If $M$ is popular in both $I$ and $I^{\prime}$, then it is popular in $\mathcal{H}_e$.
\end{lemma}

\begin{proof}
Assume $M$ is popular in both the input instances, $I$ and $I^{\prime}$. Let $z \in (A \cup B) \setminus \{a\}$. Then we can argue like previously that $\operatorname{vote}^{\mathcal{H}_e}_z(M,M')=\operatorname{vote}^I_z(M,M')=\operatorname{vote}^{I^{\prime}}_z(M,M')$.

For any matching $M'$, consider the vote of $a$ in $\mathcal{H}_e$. Consider when $\text{vote}_a^{\mathcal{H}_e}(M, M') = 1$, then $\text{vote}_a^{\mathcal{H}_e}(M, M') \ge \text{vote}_a^{I}(M, M') \ge 0$. If $\text{vote}_a^{\mathcal{H}_e}(M, M') = 0$, then $\text{vote}_a^{\mathcal{H}_e}(M, M') = \text{vote}_a^{I}(M, M') \ge 0$ because $M'(a) = M(a)$. When $\text{vote}_a^{\mathcal{H}_e}(M, M') = -1$, then $M'(a) \in P^I \cup P^{I^{\prime}}$, and without loss of generality let's assume that $M'(x) \in P^I$. Then $\text{vote}_a^{I}(M, M') = -1$, and again $\Delta^{\mathcal{H}_e}(M, M') = \Delta^{I}(M, M') \geq 0$. So, we come to this conclusion that $M$ is popular in $\mathcal{H}_e$.
\end{proof}

These lemmas lead to the following corollary:

\begin{corollary}
\label{cor:hybrid-edge-equivalence}
The instance $(I, I^{\prime})$ of \TwoSd admits a robust popular matching containing edge $e$ if and only if the hybrid instance $\mathcal{H}_e$ admits a popular matching containing $e$.
\end{corollary}

\subsection{Handling Unmatched Perturbing Agent}

We now consider the case where the agent $a$ is unmatched in a robust popular matching.

\begin{lemma}
\label{lemma:unmatched-robust}
Let $M$ be a matching in which $a$ is unmatched. Then $M$ is popular in $I$ if and only if it is popular in $I^{\prime}$.
\end{lemma}

\begin{proof}
Since $a$ is unmatched, $\text{vote}_a^{I}(M, M') = \text{vote}_a^{I^{\prime}}(M, M')$ for all $M'$. As all other agents have identical preferences, $\Delta^{I}(M, M') = \Delta^{I^{\prime}}(M, M')$.
\end{proof}

\subsection{Algorithm and Analysis}

Combining the above results, we obtain \Cref{algo:two-side-one-agent}.

\begin{algorithm}[t]
\caption{\RPMTwoSd for single agent perturbations}
\label{algo:two-side-one-agent}
\begin{algorithmic}[1]
\Require Instances $I$, $I^{\prime}$ differing only in preferences of agent $a$
\Ensure A robust popular matching or ``No'' if none exists
\State Compute a popular matching $M$ for $I$ \Comment{Use algorithm from \cite{doi:10.1137/16M1076162}}
\If{$a$ is unmatched in $M$}
    \State \Return $M$
\EndIf
\For{each edge $e$ incident to $a$}
    \If{$\mathcal{H}_e$ has a popular matching $M$ with $e \in M$} \Comment{Use algorithm from \cite{doi:10.1137/16M1076162}}
        \State \Return $M$
    \EndIf
\EndFor
\State \Return ``No robust popular matching exists''
\end{algorithmic}
\end{algorithm}

\begin{theorem}
\label{thm:one-agent-poly}
\RPMTwoSd can be solved in polynomial time when the perturbed instance differs only in the preference order of one agent.
\end{theorem}

\begin{proof}
The algorithm runs in polynomial time since there exists an algorithm that decides the popular matching problem for an instance of \TwoSd \cite{doi:10.1137/16M1076162} in $O(n^2)$ time.

For correctness: If the algorithm returns a matching at line 3, then it is popular in $I$ and leaves $a$ unmatched, so by \Cref{lemma:unmatched-robust} it is robust. If it returns at line 7, the matching is popular in $\mathcal{H}_e$ and contains $e$, so by \Cref{lem:hybrid-popular-implies-robust} it is robust.

Conversely, if a robust popular matching $M$ exists, and $M$ matches $a$ via some edge $e$, so by \Cref{cor:hybrid-edge-equivalence}, the hybrid instance $\mathcal{H}_e$ has a popular matching containing $e$, and the algorithm returns the matching at line 7 of \Cref{algo:two-side-one-agent}.
\end{proof}

\subsection{Extension to Multiple Instances}

The approach generalizes to multiple instances that differ only in the preferences of a single agent $a$. For an edge $e = \{a, b\}$, we define a generalized hybrid instance where $a$'s preference order places all agents preferred to $b$ in \emph{any} input instance above $b$. This ensures that if a matching is popular in the hybrid instance, it remains popular in all original instances.

\begin{theorem}
\label{thm:multiple-instances}
There exists a polynomial-time algorithm for determining whether a matching exists that is popular in all instances $\mathcal{I}_1, \dots, \mathcal{I}_k$, where all instances share the same underlying graph and differ only in the preferences of a single agent.
\end{theorem}

However, this approach does not extend to perturbations by multiple agents, as the correspondence between hybrid instances and robust popular matchings breaks down even when two agents swap preferences over two adjacent agents.

\section{Stable Matchings with Two-Sided Ties and Approximate Popularity}
\label{sec:stable-mat-2-sided-tie}
We are given as input a bipartite graph $G = (A \cup B, E)$ where vertices in $A$ stand for agents and the vertices included in $B$ are called the jobs. Both agents and jobs allow weak preferences, i.e., ties are allowed in their preference rankings. Some agents and jobs impose strict preference rankings over their neighbours in $G$. The model where every applicant $a \in A$ has a strict ordering of jobs and weak preferences, i.e., ties are allowed only in $b \in B$, is extensively studied \cite{10.1007/s00453-024-01215-6,doi:10.1137/16M1076162}.

We consider the stable matching problem in bipartite graphs $G = (A \cup B, E)$ where \textit{both} sides have weak preferences (ties allowed), with the restriction that every tie has length at most $k$. To the best of our knowledge, no positive results on \emph{near-popular} matchings in the setting of two-sided ties are currently known. We present a polynomial-time algorithm that computes a stable matching with \emph{unpopularity factor} at most $k$, extending Kavitha's result \cite{10.1007/s00453-024-01215-6} for one-sided ties to the two-sided ties case.

\subsection{Our Algorithm}

\begin{algorithm}
\caption{Stable Matching with Two-Sided Ties}
\label{algo:stbl-mat-2-tie}
\begin{algorithmic}[1]
\Require Bipartite graph $G = (A \cup B, E)$ with weak preferences on both sides, all ties of length $\leq k$
\Ensure Stable matching $M$ with $u(M) \leq k$
\State
\Function{StableMatchingTwoSided}{$G, k$}
    \State \textbf{Step 1: Extended Graph Construction}
    \State $A_0 \gets A \cup \{d'(b) : b \in B\}$ \Comment{Dummy agents}
    \State $B_0 \gets B \cup \{d(a) : a \in A\}$ \Comment{Dummy jobs}
    \State $E_0 \gets E \cup \{(a, d(a)) : a \in A\} \cup \{(d'(b), b) : b \in B\}$
    \For{$a \in A_0$}
        \If{$a$ has tied preferences}
            \State Break ties arbitrarily to create strict preferences
        \EndIf
    \EndFor
    \State Dummy vertices are least preferred by their real counterparts
    \State
    \State \textbf{Step 2: Proposal Phase}
    \State $G' \gets (A_0 \cup B_0, \emptyset)$
    \State $M_0 \gets \emptyset$
    \While{$\exists$ even vertex $a \in A_0$ with $\deg_{G'}(a) < k$}
        \State \Call{Propose}{$a$} \Comment{Using Kavitha's subroutine \cite{10.1007/s00453-024-01215-6}}
    \EndWhile
    \State $M_0 \gets$ maximum matching in $G'$
    \State
    \State \textbf{Step 3: Final Matching}
    \State $M \gets M_0 \setminus \{\text{edges incident to dummy vertices}\}$
    \State \Return $M$
\EndFunction
\end{algorithmic}
\end{algorithm}

The \textsc{Propose} subroutine is identical to Kavitha's algorithm \cite{10.1007/s00453-024-01215-6}:
\begin{itemize}
\item Agent $a$ proposes to neighbors in decreasing order of preference
\item Jobs break ties based on $\text{rank}_{G'}(a,b)$ values
\item The rank represents the position of $b$ in $a$'s current preference list within $G'$
\end{itemize}

\subsection{Proof of Correctness}
We divide the proof of correctness of \Cref{algo:stbl-mat-2-tie} into two parts, e.g., stability of the returned matching $M$ and bound of the \emph{unpopularity factor}. Then we discuss the running time.
\subsection{Stability}

\begin{theorem}
\label{thm:ret-algo-stbl-mat-2-tie}
The matching $M$ returned by \Cref{algo:stbl-mat-2-tie} is stable in $G$.
\end{theorem}

\begin{proof}
Assume for contradiction that $(a,b)$ blocks $M$ in $G$. Then:
\begin{itemize}
\item $a$ prefers $b$ to $M(a)$ (or $a$ is unmatched and prefers $b$)
\item $b$ prefers $a$ to $M(b)$ (or $b$ is unmatched and prefers $a$)
\end{itemize}

In the extended graph $G_0$:
\begin{itemize}
\item If $a$ is unmatched in $M$, then $M_0(a) = d(a)$, and $a$ prefers $b$ to $d(a)$
\item If $b$ is unmatched in $M$, then $M_0(b) = d'(b)$, and $b$ prefers $a$ to $d'(b)$
\end{itemize}

Thus $(a,b)$ blocks $M_0$ in $G_0$. However, by the correctness of Kavitha's algorithm \cite{10.1007/s00453-024-01215-6} applied to $G_0$ (which has strict preferences on $A_0$), $M_0$ is stable in $G_0$. Contradiction.
\end{proof}

\subsection{Unpopularity Factor Bound}

\begin{definition}
For any edge $(a,b) \in E \setminus M$, define the label:
\begin{itemize}
\item $\text{vote}_a(b, M) = +1$ if $a$ prefers $b$ to $M(a)$, $-1$ if $a$ prefers $M(a)$ to $b$, $0$ if indifferent
\item $\text{vote}_b(a, M) = +1$ if $b$ prefers $a$ to $M(b)$, $-1$ if $b$ prefers $M(b)$ to $a$, $0$ if indifferent
\end{itemize}
Label each edge $(a,b) \in E \setminus M$ as $(\text{vote}_a(b, M), \text{vote}_b(a, M))$.
\end{definition}

\begin{lemma}
No edge in $E \setminus M$ is labeled $(+1, +1)$.
\end{lemma}

\begin{proof}
Such an edge would block $M$, contradicting \Cref{thm:ret-algo-stbl-mat-2-tie}.
\end{proof}

\begin{lemma}
\label{lemma:max-seg-len}
For any alternating path/cycle $\rho$ in $M \oplus N$, any maximal contiguous segment of edges labeled $(+1, 0)$ has length at most $k-1$.
\end{lemma}

\begin{proof}
Consider an alternating path:
\[
\rho = \langle \dots, e_1, f_1, e_2, f_2, \dots, e_t, f_t \rangle \quad \text{where } e_i \in M, f_i \in N
\]

Assume $f_1, f_2, \dots, f_{k-1}$ are all labeled $(+1, 0)$. We show $f_k$ cannot be labeled $(+1, 0)$.

For each $f_i = (a_i, b_{i+1})$ labeled $(+1, 0)$:
\begin{itemize}
\item $a_i$ prefers $b_{i+1}$ to $M(a_i) = b_i$
\item $b_{i+1}$ is indifferent between $a_i$ and $M(b_{i+1}) = a_{i+1}$
\end{itemize}

\begin{claim}
$\text{rank}_{G'}(a_{i+1}, b_{i+1}) < \text{rank}_{G'}(a_i, b_i)$ for $1 \leq i \leq k-1$
\end{claim}

\begin{proof}[Proof of Claim]
Case 1: $(a_i, b_{i+1}) \in G'$ \\
Then $\text{rank}_{G'}(a_i, b_{i+1}) < \text{rank}_{G'}(a_i, b_i)$ since $a_i$ prefers $b_{i+1}$ to $b_i$.

Case 2: $(a_i, b_{i+1}) \notin G'$ \\
Let $r = \text{rank}_{G'}(a_i, b_{i+1})$ from $a_i$'s last proposal to $b_{i+1}$ \\
Since $(a_i, b_{i+1}) \notin G'$ but $(a_{i+1}, b_{i+1}) \in G'$, we have $\text{rank}_{G'}(a_{i+1}, b_{i+1}) < r$

In both cases: $\text{rank}_{G'}(a_{i+1}, b_{i+1}) \leq \text{rank}_{G'}(a_i, b_{i+1}) < \text{rank}_{G'}(a_i, b_i)$
\end{proof}

Thus we have a chain:
\[
\text{rank}_{G'}(a_k, b_k) < \text{rank}_{G'}(a_{k-1}, b_{k-1}) < \cdots < \text{rank}_{G'}(a_1, b_1) \leq k
\]

Therefore $\text{rank}_{G'}(a_k, b_k) \leq 1$, so $b_k$ is $a_k$'s most preferred neighbor in $G'$.

Now, if $a_k$ prefers $b_{k+1}$ to $b_k$, then:
\begin{itemize}
\item $(a_k, b_{k+1})$ must be in $G_0$ (original graph)
\item $a_k$ would have proposed to $b_{k+1}$ before $b_k$ with rank 1
\item Since $b_{k+1}$ didn't accept, it must prefer its current partner to $a_k$
\item Thus $\text{vote}_{b_{k+1}}(a_k, M) = -1$
\end{itemize}

Therefore $f_k = (a_k, b_{k+1})$ is labeled $(+1, -1)$, not $(+1, 0)$.
\end{proof}

\begin{lemma}
In any alternating cycle, not all edges in $\rho \setminus M$ can be labeled $(+1, 0)$.
\end{lemma}

\begin{proof}
Follows from \Cref{lemma:max-seg-len} by considering the cycle as a path.
\end{proof}

\begin{theorem}
The matching $M$ has unpopularity factor $u(M) \leq k$.
\end{theorem}

\begin{proof}
Let $N$ be any matching in $G$. Consider $M \oplus N$ decomposed into alternating paths and cycles.

For each alternating component $\rho$:
\begin{itemize}
\item Extract maximal contiguous segments $s_1, \dots, s_r$ of edges in $\rho \cap N$ labeled $(+1, 0)$
\item By \Cref{lemma:max-seg-len}, $|s_i| \leq k-1$ for each $i$
\item Each $s_i$ is followed by an edge in $\rho \cap N$ with at least one ``$-1$'' in its label
\item Map all edges in $s_i$ to this following edge
\end{itemize}

Each edge with a ``$-1$'' vote is mapped to by at most $k-1$ edges with $(+1, 0)$ labels.

Now count votes:
\begin{itemize}
\item $\phi(N, M)$ = total $+1$ votes from edges in $N \setminus M$
\item $\phi(M, N)$ = total $-1$ votes from edges in $N \setminus M$ (plus potential extra from $M \setminus N$)
\end{itemize}

From our mapping:
\[
(\text{\# of $+1$ votes from $(+1, 0)$ edges}) \leq (k-1) \cdot (\text{\# of edges with $-1$ votes})
\]

Also, edges labeled $(+1, -1)$ contribute one $+1$ and one $-1$.

Therefore:
$\phi(N, M) \leq (k-1) \cdot \phi(M, N) + \phi(M, N) = k \cdot \phi(M, N)$
\end{proof}

\subsection{Running Time}

\begin{theorem}
The algorithm runs in $O(k^2 \cdot m \cdot n)$ time.
\end{theorem}

\begin{proof}
Graph $G_0$ has $O(n)$ vertices and $O(m)$ edges. Each \textsc{Propose}$(a)$ call takes $O(k \cdot n)$ time. Total \textsc{Propose} calls takes time $O(k \cdot m)$ across all agents. Maintaining maximum matching and Dulmage-Mendelsohn decomposition takes $O(m)$ per update. Which translates to total time $O(k^2 \cdot m \cdot n)$
\end{proof}
\Cref{algo:stbl-mat-2-tie} successfully extends Kavitha's result to two-sided ties while maintaining stability of the output matching, unpopularity factor $\leq k$
and polynomial running time $O(k^2 \cdot m \cdot n)$. This resolves the open problem from the original paper and provides a complete solution for stable matchings with two-sided ties and bounded unpopularity factor.

\section{Conclusion and Open Questions} 
This paper investigates the existence and computation of robust popular matchings in bipartite graphs under two preference models: the one-sided model and the two-sided model with one-sided ties. We provide polynomial-time algorithms to decide whether a robust popular matching exists when the input instances differ only in the preference order of a single agent. For the one-sided model, we leverage known characterizations of popular matchings and analyse specific perturbation cases. For the two-sided model with one-sided ties, we introduce the concept of hybrid instances and establish a correspondence between robust popular matchings and popular matchings in these hybrid instances. Additionally, the paper presents an algorithm for computing stable matchings with bounded unpopularity factors in the presence of two-sided weak preferences (ties), extending prior work to a more general setting. These results enhance our understanding of robustness in matching markets and provide efficient solutions for scenarios with limited preference uncertainty.

Several open questions remain for future work. First, can the robustness results be extended to cases where more than one agent changes their preferences? We note that our approach for the two-sided model does not generalize to multiple perturbing agents. Second, is it possible to design efficient algorithms for robust popular matchings in models with more general preference structures, two-sided weak preferences without bounded tie lengths? Third, what is the computational complexity of finding robust popular matchings when the underlying graph changes across instances, rather than just the preferences? Finally, can similar robustness guarantees be achieved for other solution concepts in matching, such as stability or Pareto optimality, under similar perturbation models?


\bibliography{sample}

\end{document}